\documentclass[11pt,a4paper]{article}
\usepackage{microtype}
 \usepackage{ifpdf}
 \ifpdf
 \usepackage[pdftex]{graphicx}
 \else
\usepackage{graphicx}
 \fi
\usepackage{booktabs}
\usepackage[T1]{fontenc}
\usepackage{textcomp}
\usepackage[sc,noBBpl]{mathpazo}
\usepackage{amsmath,amsfonts,amssymb,amsthm}
\usepackage[mathscr]{eucal}
\newcommand\mvector{\boldsymbol}

\newcommand\vc{\mvector{c}}

\newcommand\vp{\mvector{p}}
\newcommand\vq{\mvector{q}}
\newcommand\vr{\mvector{r}}

\newcommand\vgamma{\mvector{\gamma}}

\newcommand\field{\mathbb}

\newcommand\R{\field{R}}
\newcommand\bbS{\mathbb{S}}
\newcommand\EE{\field{E}}
\newcommand\C{\field{C}}

\newcommand\N{\field{N}}

\newcommand\bbH{\field{H}}

\renewcommand\Re{\operatorname{Re}}
\renewcommand\Im{\operatorname{Im}}

\newcommand\tr{\operatorname{Tr}}

\newcommand\grad{\operatorname{grad}}

\newcommand\rmd{\mathrm{d}}

\newcommand\rmi{\mathrm{i}\mspace{1mu}}

\newcommand\Dt{\frac{\mathrm{d}\phantom{t} }{\mathrm{d}\mspace{1mu} t}}

\newcommand\pder[2]{\dfrac{\partial #1 }{\partial #2}}
\newcommand\abs[1]{\lvert #1 \rvert}

\newcommand\sk{\operatorname{S}_{\kappa}}
\newcommand\ck{\operatorname{C}_{\kappa}}

\newcommand\mtext[1]{\quad\text{#1}\quad}

\usepackage{natbib}
\theoremstyle{plain}
\newtheorem{theorem}{Theorem}
\newtheorem{lemma}[theorem]{Lemma}

\newtheoremstyle{note}{\topsep}{\topsep}{\slshape}{}{\scshape}{}{ }{}
\theoremstyle{note}

\numberwithin{equation}{section}
\numberwithin{theorem}{section}
\begin{document}
\thispagestyle{empty}
\vspace*{1em}
\begin{center}
\LARGE\textbf{Necessary conditions for classical super-integrability of a certain family of  potentials in constant curvature spaces }
\end{center}
\vspace*{0.5em}
\begin{center}
\large Andrzej J.~Maciejewski$^1$, Maria Przybylska$^{2,3}$ and Haruo
Yoshida$^4$
\end{center}
\vspace{2em}
\hspace*{2em}\begin{minipage}{0.8\textwidth}
\small
$^1$J.~Kepler Institute of Astronomy,
  University of Zielona G\'ora,
  Licealna 9, \\\quad PL-65--417 Zielona G\'ora, Poland
  (e-mail: maciejka@astro.ia.uz.zgora.pl)\\[0.5em]
$^2$Toru\'n Centre for Astronomy,
  N.~Copernicus University,
  Gagarina 11,\\\quad PL-87--100 Toru\'n, Poland,
  (e-mail: Maria.Przybylska@astri.uni.torun.pl)\\[0.5em]
$^3$Institute of Physics,
  University of Zielona G\'ora,
  Licealna 9, \\\quad PL-65--417 Zielona G\'ora, Poland
  \\[0.5em]
$^4$National Astronomical Observatory, 2-21-1 Osawa,  Mitaka, 181-8588 Tokyo,
Japan, (e-mail:
h.yoshida@nao.ac.jp)
\end{minipage}\\[2.5em]
   \centerline{Version of: \today}
\vspace*{1em}

\noindent {\small \textbf{Abstract.}  We formulate the necessary
  conditions for the maximal super-integrability of a certain family
  of classical potentials defined in the constant curvature two-dimensional
  spaces. We give examples of  homogeneous potentials of degree
  $-2$ on $\EE^2$ as well as their equivalents on $\bbS^2$ and
  $\bbH^2$ for which these necessary conditions are also sufficient.
  We show explicit forms of the additional first integrals which
  always can be chosen polynomial with respect to the momenta and which
  can be of an arbitrary high degree with respect to the momenta.  }

{\textbf{Key words:} super-integrable systems,  integrability,
 Hamiltonian equations, differential Galois integrability obstructions.}

{\textbf{PACS numbers:}  02.30.Ik; 45.20.Jj;  02.40.Ky; 45.50.Jf}
\section{Introduction}
In this paper we consider classical Hamiltonian systems with $n$
degrees of freedom given by Hamiltonian function $H(\vq, \vp)$, where
$\vq=(q_1, \ldots, q_n)$ are canonical coordinates and $\vp=(p_1,
\ldots, p_n)$ are the canonical momenta. We say that such a system is
maximally super-integrable if the Hamilton's equations of motion
\begin{equation}
  \label{eq:hem}
  \Dt q_i=\pder{H}{p_i}(\vq, \vp), \qquad \Dt p_i=-\pder{H}{q_i}(\vq,
  \vp), \qquad i=1,\ldots, n,
\end{equation}
admit $2n-1$ functionally independent first integrals such that among
them $n$ commute.  Although such systems are highly exceptional they
attract a lot of attentions. Motivation for study of super-integrable
systems comes from the classical as well as from quantum physics, see,
e.g., \cite{MR2105429}. Classical maximal integrability implies that
all bounded trajectories are closed and the motion is periodic.  In
quantum mechanics maximal super-integrability means the existence of
$2n-1$ well defined, algebraically independent operators (including
Hamiltonian) among them $n$ pairwise commute see
e.g. \cite{Gravel:02::,Rodriguez:02::} and references therein.  For
quantum systems the maximal super-integrability implies the degeneracy
of energy levels.  The problem of construction of integrable,
super-integrable and maximally super-integrable quantum systems from
the corresponding classical ones is very complicated, see
e.g. \cite{Hietarinta:84::,Gravel:02::}, and it will be not considered
in this paper.

Recently, some remarkable families of super-integrable systems were
found.  In \cite{1751-8121-42-24-242001,2010JPhA...43a5202T} the
authors introduced a family of quantum and classical systems for which
the classical Hamiltonian function in polar coordinates $(r,\varphi)$
is given by
\begin{equation}
  \label{eq:hnf}
  H_n^{(0)}= \frac{1}{2}\left(p_r^2 +\frac{p_{\varphi}^2}{r^{2}} \right)
  +V_n^{(0)}(r,\varphi),
\end{equation}
where the potential $V_n^{(0)}(r,\varphi)$ has the form
\begin{equation}
  \label{eq:V0}
  V^{(0)}_n(r,\varphi) :=\frac{a}{r^2\cos^2(n\varphi)}+ 
  \frac{b}{r^2\sin^2(n\varphi)},
\end{equation}
$n$ is an integer, $a$, and $b$ are parameters\footnote{In the cited
  paper the potential function contains the harmonic oscillator term
  $\omega^2r^2/2$. }.  The system is integrable and it has the
following first integral
\begin{equation}
  \label{eq:gic}
  G:=\frac{1}{2} p_{\varphi}^2+r^2V_n^{(0)}(r,\varphi).
\end{equation}
In quantum version momenta are replaced by appropriate partial
derivative operators.  As it was shown in
\cite{1751-8121-42-24-242001}, for small integer values of $n$ the
quantum system is super-integrable, and moreover the degree with
respect to the momenta of the second additional first integral grows
with $n$. On this basis a conjecture that the system is
super-integrable for an arbitrary $n$ was formulated. Later it was
justified that in fact the quantum system is super-integrable for any
integer odd $n$ in \cite{1751-8121-43-8-082001} and even for rational
$n$ in \cite{Kalnins:arXiv1002.2665}. For the classical system in
\cite{2010JPhA...43a5202T} it was shown that all bounded trajectories
are closed for all integer and rational values of $n$ and
in~\cite{1751-8121-43-9-092001} forms of first integrals were
described.

The second example needs a more detailed presentation. We consider a
point with the unit mass moving on a sphere $\bbS^2\subset \R^3$ under
influence of Hooke forces.  A Hooke centre located at point
$\vr\in\bbS^2$ generates a potential field of forces. The value of the
potential at point $\vgamma\in\bbS^2$ is
\begin{equation}
  \label{eq:sh}
  V=\frac{\alpha}{(\vgamma\cdot\vr)^2},
\end{equation}
where $\alpha$ is the intensity of the centre.  In \cite{Borisov:09::}
the authors investigated the problem of a point mass moving on a
sphere $\bbS^2$ in the field of odd number $n=2l+1$ of Hooke centres
with equal intensities located in a great circle at the vortexes of
the regular $n$-gon.  We can assume that the Hooke centres are located
at points
\begin{equation}
  \label{eq:nc}
  \vr_{k;n}=\left(   \sin \varphi_{k;n} , \cos\varphi_{k;n}, 0 \right), 
  \mtext{where}\varphi_{k;n}:= \frac{2\pi k}{n} \mtext{for}
  k=1,\ldots, n.
\end{equation}
Thus, the potential has the form
\begin{equation}
  \label{eq:vnk}
  V_n^{(1)}:=     \sum_{k=1}^n\frac{\alpha}{\sin^2\theta
    \sin^2\left(\varphi+\varphi_{k;n}\right)}=
  \frac{\alpha n^2}{\sin^2\theta \sin^2n\varphi}. 
\end{equation}
In the above we used the following trigonometric identity
\begin{equation*}
  \sum_{k=1}^n\frac{1}{
    \sin^2\left(\varphi+\varphi_{k;n}\right)}=
  \frac{n^2}{ \sin^2n\varphi},
\end{equation*}
see e.g. \cite{Jakubsky2005154}.

Thus denoting $a:=\alpha n^2$, we can write the Hamiltonian of the
system in the following form
\begin{equation}
  \label{eq:hs1}
  H_n^{(1)}=\frac{1}{2}\left(p_{\theta}^2+
    \frac{p_{\varphi}^2}{\sin^2\theta}\right)+
  \frac{a}{\sin^2\theta \sin^2 n\varphi}. 
\end{equation}
The system is integrable and the following function
\begin{equation}
  \label{eq:fis}
  F_1:= \frac{1}{2}
  p_{\varphi}^2+
  \frac{a}{\sin^2 n\varphi}
\end{equation}
is a first integral. In \cite{Borisov:09::} it was shown that this
Hamiltonian system is super-integrable as it possesses the second
additional first integral of degree either $2n+1$ or $2n+2$ with
respect to the momenta.  The proof of this fact is remarkably natural
and simply. We show later that it allows to give several
generalisations of the above two examples.

The above examples have the same, in some sense nature.  In Cartesian
coordinates potential~\eqref{eq:V0} has the form
\begin{equation}
  V_n^{(0)}=\dfrac{a(q_1^2+q_2^2)^{n-1}}{[\operatorname{Re}(q_1+\rmi q_2)^n]^2}+
  \dfrac{b(q_1^2+q_2^2)^{n-1}}{[\operatorname{Im}(q_1+\rmi q_2)^n]^2}.
  \label{eq:complet}
\end{equation}
Hence it is a rational homogeneous function of degree $-2$.  As it is
well known a natural Hamiltonian system given by
\begin{equation}
  \label{eq:he2} 
  H^{(0)}:=\frac{1}{2}(p_1^2+p_2^2) +V(q_1,q_2),
\end{equation}
with a homogeneous potential of degree $-2$ is integrable because it
possesses the following first integral
\begin{equation}
  \label{eq:fie2}
  F_1:= \frac{1}{2}(q_1p_2-q_2p_1)^2 +(q_1^2+q_2^2)V(q_1,q_2).
\end{equation}
For a point on a sphere we have analogous potential. Namely,
Hamiltonian system given by
\begin{equation}
  \label{eq:hs}
  H_n^{(1)}=\frac{1}{2}\left(p_{\theta}^2+
    \frac{p_{\varphi}^2}{\sin^2\theta}\right)+
  \frac{1}{\sin^2\theta}U( \varphi), 
\end{equation}
is integrable with the first integral
\begin{equation}
  \label{eq:fis1}
  F_1:= \frac{1}{2}
  p_{\varphi}^2+
  U(\varphi).
\end{equation}

Considering the above two examples, we can ask whether a potential of
the prescribed form is super-integrable. It appears that this question
is difficult if we look for an effective and computable necessary
conditions for the super-integrability.

In this paper we consider natural Hamiltonian systems with two degrees
of freedom defined on $T^{\star}M$ where $M$ is a two dimensional
manifold with a constant curvature metrics. More specifically, $M$ is
either sphere $\bbS^2$, Euclidean plane $\EE^2$, or the hyperbolic
plane $\bbH^2$.  In order to consider those three cases simultaneously
we will proceed as in \cite{Herranz:00::,Ranada:99::a} and we define
the following functions
\begin{equation}
  \label{eq:ck}
  \ck(x):= 
  \begin{cases} 
    \cos(\sqrt{\kappa}x) &\text{for} \quad\kappa>0, \\
    1 & \text{for}\quad  \kappa=0, \\
    \cosh(\sqrt{-\kappa}x)&\text{for} \quad \kappa<0,
  \end{cases}
\end{equation}
\begin{equation}
  \label{eq:sk}
  \sk(x):= 
  \begin{cases} 
    \frac{1}{\sqrt{\kappa}}\sin(\sqrt{\kappa}x) &\text{for} \quad\kappa>0, \\
    x & \text{for} \quad\kappa=0, \\
    \frac{1}{\sqrt{-\kappa}}\sinh(\sqrt{-\kappa}x)&\text{for}\quad
    \kappa<0.
  \end{cases}
\end{equation}
These functions satisfy the following identities
\begin{equation}
  \ck^2(x)+\kappa \sk^2(x)=1,\quad \sk'(x)=\ck(x),\quad \ck'(x)=-\kappa\sk(x).
  \label{eq:trigono}
\end{equation}
We consider natural systems Hamiltonian systems with potential
$V(r,\varphi)$ defined by
\begin{equation}
  \label{eq:hamkappa}
  H^{(\kappa)}= \frac{1}{2}\left(p_r^2 + \frac{p_{\varphi}^2}{\sk^2(r)  }\right) +V(r,\varphi).
\end{equation}
The form of the kinetic energy corresponds to the metric on $M$ with
constant curvature $\kappa$.  Our aim is to distinguish a special
class of super-integrable potentials. Inspired by the examples
discussed above we consider potentials of the form
\begin{equation}
  \label{eq:v}
  V^{(\kappa)}(r, \varphi):= \frac{1}{\sk^2(r)}U(\varphi). 
\end{equation}
These potentials are separable.  In fact
\begin{equation}
  \label{eq:int}
  G:=\frac{1}{2} p_{\varphi}^2 + U(\varphi),
\end{equation}
is a first integral of the system and we have also
\begin{equation}
  \label{eq:hg}
  H=\frac{1}{2} p_{r}^2 + \frac{1}{\sk^2(r)}G.
\end{equation}
In order to formulate our main result let us assume that there exists
$\varphi_0\in\C$ such that $U'(\varphi_0)=0$ and
$U(\varphi_0)\neq0$. Under this assumption we define the following
quantity
\begin{equation}
  \label{eq:lam}
  \lambda:= 1-\frac{1}{2}\frac{U''(\varphi_0)}{U(\varphi_0)}.
\end{equation}
The most important result of this paper is formulated in the following
theorem which gives necessary conditions for the super-integrability
of systems~\eqref{eq:hg} with potential~\eqref{eq:v}.
\begin{theorem}
  \label{thm:we}
  Assume that potential $V^{(\kappa)}$ given by~\eqref{eq:v} satisfies
  the following assumption: there exists $\varphi_0\in\C$ such that
  $U'(\varphi_0)=0$ and $U(\varphi_0)\neq0$.  If $V^{(\kappa)}$ is
  super-integrable, then
  \begin{equation*}
    \lambda:= 1-\frac{1}{2}\frac{U''(\varphi_0)}{U(\varphi_0)}=1-s^2,
  \end{equation*}
  for a certain non-zero rational number $s$.
\end{theorem}

The necessary conditions for the super-integrability given by
Theorem~\ref{thm:we} are deduced from an analysis of the differential
Galois group of the variational equations along the described
particular solution. Here we refer to our
paper~\cite{Maciejewski:08::c} where the reader will find a
description of applications of the differential Galois theory to a
study of the integrability and the super-integrability as well as an
analysis of the case $\kappa=0$.  Indeed the statement of
Theorem~\ref{thm:we} for the case $\kappa=0$ is just a rephrase of the
previous result written in polar coordinates, as it will be seen in
the next section.

\section{Relation with known necessary conditions for super-integrability}
For a systems   on $\EE^2$  given by natural Hamiltonian~\eqref{eq:he2} 
with a homogeneous potential $V(q_1,q_2)$ of degree $k$ equations of
motion have the form
\begin{equation}
\label{eq:he}
\Dt q_i =p_i, \qquad \Dt p_i=-\pder{V}{q_i}, \quad i=1, 2.
\end{equation}
One can look for their particular solution of the form
$\vq(t)=\varphi(t) \vc $, $\vp(t)=\dot\varphi(t) \vc $, where
$\vc\in\C^2$ is a non-zero vector, and $\varphi(t)$ is a scalar
function. As it is easy to see such a solution exists provided that
$\vc$ is a non-zero solution of $\grad V(\vc)=\vc$, and $\varphi(t)$
satisfies $\ddot \varphi + \varphi^{k-1}=0$.  Vector $\vc$ is called
the Darboux point of potential $V$. Then the necessary conditions for
the integrability and the super-integrability which come from an
analysis of the differential Galois group of the variational equations
along the described particular solution are expressed by means of one
eigenvalue of the Hessian matrix $V''(\vc)$, see \cite{Morales:99::c}
for the integrability obstructions, and \cite{Maciejewski:08::c} for
super-integrability conditions. For matrix $V''(\vc)$ vector $\vc$ is
an eigenvector with the corresponding eigenvalue $(k-1)$. Thus the
other eigenvalue is given by
\begin{equation}
\lambda = \tr V''(\vc) - (k-1) = \nabla^2 V(\vc) - (k-1).
\end{equation}
The mentioned above necessary conditions for the integrability  have
the  form of arithmetic  restrictions imposed on $\lambda$. For
example, in our previous paper \cite{Maciejewski:08::c}, we proved,
among other things the following. 
\begin{theorem}
\label{thm:we2008}
If  a Hamiltonian system given by~\eqref{eq:he2},
with a homogeneous potential $V(q_1,q_2)$ of degree $k$, $\abs{k}\leq
2$  is super-integrable, then for each Darboux point $\vc$  the
corresponding eigenvalue $\lambda$ satisfies  the following
conditions:
\begin{itemize}
\item if $k=2$, then $\lambda=s^2$, where $s$ is a non-zero rational
  number;
\item if $k=1$, then $\lambda=0$;
\item if $k=-1$, then $\lambda=1$;
 \item if $k=-2$, then $\lambda=1-s^2$, where $s$ is a non-zero rational
  number.
\end{itemize}
\end{theorem}
In the polar coordinates homogeneous potential have  the form
\begin{equation}
\label{eq:pw}
V(\vq)= V(r\cos\varphi, r\sin\varphi) = r^k U(\varphi),
\end{equation}
and a Darboux point  is given by 
\begin{equation}
\label{eq:c}
(c_1,c_2) = c(\cos\varphi_0, \sin\varphi_0),  
\end{equation}
where $\varphi_0$ is a solution of $U'(\varphi)=0$ such that
$U(\varphi_0)\neq 0$. 

The  Laplacian $\nabla^2 V$ of function $V(\vq)$  takes the form
\[
\nabla^2 V = \dfrac{\partial^2 V}{\partial q_1^2}+ \dfrac{\partial^2 V}{\partial q_2^2}
= \dfrac{1}{r} \left[ \dfrac{\partial}{\partial r} \left( r \dfrac{\partial V}{\partial r}\right) \right] + \dfrac{1}{r^2} \dfrac{\partial^2 V}{\partial \varphi  ^2}
\]
in polar coordinates, and for $V = r^k U(\varphi  )$,
\[
\nabla^2 V = k^2 r^{k-2} U(\varphi  ) + r^{k-2} U''(\varphi  ) .
\]
Thus, as computed by \cite{0305-4470-30-16-026}, we have
\begin{equation}
 \lambda = k^2 c^{k-2} U(\varphi  _0) + c^{k-2} U''(\varphi  _0) - (k-1)
 = 1 + c^{k-2} U''(\varphi  _0)
 = 1 + \dfrac{U''(\varphi  _0)}{k \; U(\varphi  _0)},
\end{equation} 
and substitution $k = -2$ reproduces $\lambda$ in
~\eqref{eq:lam}. Therefore the statement of Theorem~\ref{thm:we2008}
with $k=-2$ gives Theorem~\ref{thm:we} immediately.  So the novelty of
Theorem~\ref{thm:we} is to confirm that the same statement holds,
independent of the value of the curvature $\kappa$.

\section{Proof of Theorem~\ref{thm:we}}
According to Theorem~1.2 in~\cite{Maciejewski:08::c}, if the
considered system is maximally super-integrable, then the identity
component the differential Galois group of the normal variational
equations along a particular solution is just the identity.

The assumptions in Theorem~\ref{thm:we} guarantee that 
\begin{equation}
\label{eq:ps}
t\longmapsto(\varphi_0, r(t),0,\dot r(t)) 
\end{equation}
is a particular solution of the system provided that $r(t)$ satisfies 
\begin{equation}
\label{eq:ddotr}
\ddot r = 2 U(\varphi_0)\frac{\ck(r)}{\sk(r)^3}.
\end{equation}
We consider a particular solution with the energy $e$, i.e., we fix 
\begin{equation}
\label{eq:e}
e = \frac{1}{2}{\dot r}^2 + \frac {U(\varphi_0)}{\sk(r)^2}.
\end{equation}

Variational equations  along particular solution~\eqref{eq:ps} have  the  form
\begin{equation}
 \begin{bmatrix}
  \dot R\\
\dot \Phi\\
\dot P_R\\
\dot P_{\Phi}
 \end{bmatrix}=
\begin{bmatrix}
 0&0&1&0\\
0&0&0&\sk^{-2}(r)\\
2\sk^{-4}(r)[2\kappa \sk^2(r)-3]U(\varphi_0)&0&0&0\\
0&-\sk^{-2}(r)U''(\varphi_0)&0&0
\end{bmatrix}
\begin{bmatrix}
  R\\
 \Phi\\
 P_R\\
 P_{\Phi}
 \end{bmatrix}.
\end{equation}   
Since the motion takes place in the plane $(r,p_r)$ the normal part of variational equations is
\[
 \dot \Phi=\sk^{-2}(r)P_{\Phi},\qquad \dot
 P_{\Phi}=-\sk^{-2}(r)U''(\varphi_0)\Phi.
\]
We rewrite this system as  one equation of the second order
\[
\ddot \Phi+a(r,p_r)\dot \Phi+b(r,p_r)\Phi=0,
\]
where 
\[
 a(r,p_r)=2\frac{\ck(r)}{\sk(r)}p_r,\quad b(r,p_r)=\sk^{-4}(r)U''(\varphi_0).
\]
Making the following change of the independent variable  $t\mapsto
s=\sk(r)$, we  transform this  equation into a  linear equation with rational coefficients
\begin{equation}
\label{eq:Phiz}
 \Phi''+p(s)\Phi'+q(s)\Phi=0,
\end{equation}
where
\begin{equation}
\begin{split}
&p(s)=\dfrac{\ddot s +\dot s a}{\dot
  s^2}=\frac{1}{s}+\frac{es}{es^2-B}+\frac{\kappa s}{\kappa s^2-1},
\\
&q(s)=\dfrac{b}{\dot s^2}=\dfrac{B ( \lambda-1)}{s^2 (e s^2 - B ) (\kappa s^2 -1)},\qquad  B:= U (\varphi_0).
\end{split}
\end{equation} 
Finally, we perform  one more transformation of the  independent variable  putting 
\begin{equation}
\label{eq:y}
z=\frac{e s^2-B}{(e-\kappa B)s^2}.
\end{equation}
After this transformation equation~\eqref{eq:Phiz} reads  
\begin{equation}
\label{eq:Phiy}
\frac{\rmd^2 \Phi}{\rmd z^2} + P\frac{\rmd \Phi}{\rmd z} +Q\Phi =0, 
\end{equation}
where 
\begin{equation}
\label{eq:PQ}
P=\frac{2z-1}{2z(z-1)}, \qquad Q:= \frac{\lambda-1}{4z(z-1)}.
\end{equation}
This is the Gauss hypergeometric equation  for which the differences of
exponents at $z=0$, $z=1$ and $z=\infty$ are 
\begin{equation}
\label{eq:dex}
\rho=\frac{1}{2}, \qquad \sigma=\frac{1}{2}, \qquad \tau = \sqrt{1-\lambda},
\end{equation}
respectively.  Putting 
\begin{equation}
\label{eq:tor}
w=\Phi \exp\left[\frac{1}{2}\int P(\zeta) \rmd \zeta\right],
\end{equation}
 we transform this  equation to its reduced form 
\begin{equation}
\label{eq:rf}
\frac{\rmd^2 w}{\rmd z^2}=r(z) w,
\end{equation}
where  
\begin{equation}
\label{eq:r}
r(z)=- \frac{4\lambda z (z-1)+3}{16z^2(z-1)^2}.
\end{equation}
This equation coincides with equation (A.9) in
\cite{Maciejewski:08::c} in which we substitute $k=-2$. Thus, we can
apply Proposition~A.3 and Proposition~A.4
from~\cite{Maciejewski:08::c} to equation~\eqref{eq:rf}, and this
exactly gives the thesis of our theorem.
\section{Tremblay-Turbiner-Winternitz (TTW) system: A family of super-integrable systems in Euclidean 
flat space $\EE^2$}

\subsection{Finding the potential by use of the necessary conditions for super-integrability}  

In general, values of $\lambda$ in \eqref{eq:rf} can depend on the
angle $\varphi_0$, or, on the Darboux point, and in order to be
super-integrable, the only necessary condition is that $\lambda = 1 -
s^2$, for a certain non-zero rational number $s$. However, there is a
special class of super-integrable potentials, for which $\lambda$
takes the same admissible value for all choices of the angle
$\varphi_0$. Indeed, the known super-integrable
potential~\eqref{eq:V0} is strongly characterised by possessing this
property as it is shown below.

First, condition $\lambda = 1 - s^2$ with expression \eqref{eq:lam} gives
\begin{equation}
-\dfrac{1}{2} \dfrac{U''(\varphi  _0)}{ \; U(\varphi  _0)} = - s^2 .
 \label{si}
\end{equation} 
Let us change the dependent variable $U(\varphi  )$, or the angular part of the potential, by
\[
U(\varphi  ) =\dfrac{1}{[f(\varphi  )]^2}.
\]
Then at a point where $U'(\varphi  _0) = 0$, we have  $f'(\varphi  _0) = 0$, and furthermore,
\begin{equation}
- \dfrac{1}{2}\dfrac{U''(\varphi  _0)}{ U(\varphi  _0)} =  \dfrac{f''(\varphi  _0)}{f(\varphi  _0)}.
\end{equation} 
Next we force the additional requirement that $\lambda$ takes the same
value $\lambda=1-s^2$, for all Darboux points $U'(\varphi  _0) = 0$.
That is, we have the relation
\begin{equation}
\label{eq:eqf}
f''(\varphi  _0) = - s^2 f(\varphi  _0),
\end{equation}
for all $\varphi _0$ such that $f'(\varphi _0)$ vanishes. Here $s$ is
a non-zero rational number, which is independent of $\varphi _0$. This
requirement does not determine function $f$ uniquely. However, we can
find examples of $f$ satisfying this condition. The simplest and naive
example is the function $f$ which satisfies the differential equation
\begin{equation}
f''(\varphi) = - s^2 f(\varphi).
\end{equation}
This is equivalent to assume that $f$ satisfies~\eqref{eq:eqf} for all  $\varphi
_0\in(0,2\pi)$. Then we find two independent solutions for $f$
\[
 f_1(\varphi  ) = \cos (s \varphi  ), \quad f_2(\varphi  ) = \sin (s \varphi  )
\]
and therefore
\begin{equation}
 U_1(\varphi  ) = \dfrac{1}{\cos^2 (s \varphi  )}, \quad 
 U_2(\varphi  ) = \dfrac{1}{\sin^2 (s \varphi  )}.
\end{equation}
are the desired angular part of the super-integrable potential.
Note that any linear combination of $U_1(\varphi  )$ and $U_2(\varphi  )$, namely
\begin{equation}
 U = a U_1(\varphi  ) + b  U_2(\varphi  ) 
   =  \dfrac{a}{\cos^2 (s \varphi  )} +  \dfrac{b}{\sin^2 (s\varphi  )}
\end{equation}
has the same property, but unless $a b = 0$ the value of $\lambda$ is $\lambda = 1 - (2 s)^2$, instead of $\lambda = 1 - s^2$.
Indeed, for the potential~\eqref{eq:V0} we have the following
\begin{lemma}
  \label{lem:1}
  If $\varphi_0\in \C$ is a solution of equation $ U'(\varphi) = 0 $
  such that $U(\varphi_0)\neq 0$, then
  \begin{equation}
    \label{eq:lama}
    \lambda= 1-\frac{1}{2}\frac{U''(\varphi_0)}{U(\varphi_0)}=
    \begin{cases}
      1 - n^2,  &\text{if}\quad ab=0,\\
      1- (2 n)^2, & \text{otherwise}.
    \end{cases}
  \end{equation}
\end{lemma}
\begin{proof}

For the potential~\eqref{eq:V0} we have
\begin{equation}
  \label{eq:uf}
  U(\varphi) = \frac{a}{\cos^2(n\varphi)}+ \frac{b}{\sin^2(n\varphi)}, 
\end{equation}
so
\begin{equation}
  \label{eq:guf}
  U'(\varphi) = 2na\frac{\sin(n\varphi) }{\cos^3(n\varphi)}-
  2nb\frac{\cos(n\varphi)}{\sin^3(n\varphi)}, 
\end{equation}
and
\begin{equation}
  \label{eq:huf}
  U''(\varphi) = 2n^2a \frac{1+2\sin^2(n\varphi) }{\cos^4(n\varphi)}+
  2n^2b \frac{1+2\cos^2(n\varphi)}{\sin^4(n\varphi)}.
\end{equation}

  Let us assume that $ab\neq 0$. In a case when $a\neq 0$ and $b=0$,
  from~\eqref{eq:guf} we have $\sin(n\varphi_0)=0$.  Thus,
  $U(\varphi_0)=a$, and, by~\eqref{eq:huf}, $U''(\varphi_0)= 2n^2 a$,
  so we have relation~\eqref{eq:lama}. In a similar way we show that
  this formula is valid in the case $a=0$ and $b\neq0$.

  If $ab\neq0$, then from~\eqref{eq:guf} we find that
  \begin{equation*}
    b=a\tan^4(n\varphi_0). 
  \end{equation*}
  Using this relation we obtain
  \begin{equation*}
    U(\varphi_0)= \frac{a}{\cos^4(n\varphi_0)}, \mtext{and} U''(\varphi_0)= \frac{8n^2a}{\cos^4(n\varphi_0)},
  \end{equation*}
  and this finishes the proof.
\end{proof}

\subsection{Checking the super-integrability by separation of variables}
As we have shown that if $n$ is a non-zero rational number, then
potential~\eqref{eq:V0} satisfies the necessary condition for the
super-integrability, next we are going to prove that this potential is
indeed super-integrable.  We demonstrate this giving  an explicit form of the
second additional first integral $F$ which is functionally independent
together with $H_n^{(0)}$ and $G$ given by~\eqref{eq:hnf}
and~\eqref{eq:gic}, respectively.  However, in order to demonstrate
how to derive this first integral we consider at first simplified case
when $b=0$ in potential~\eqref{eq:V0}.  Then Hamiltonian takes the
form
\begin{equation}
  H=\dfrac{1}{2}\left(p_r^2+\dfrac{p_{\varphi}^2}{r^2}\right)
  +\dfrac{a}{r^2\cos^2(n\varphi)}= \dfrac{p_r^2}{2}+\dfrac{1}{r^2}G,    
\label{eq:hamidlo}
\end{equation} 
where $G$ is the first integral given by 
\begin{equation}
  G=\dfrac{p_{\varphi}^2}{2}+\dfrac{a}{\cos^2(n\varphi)}.
\end{equation} 
In order to perform the explicit integration we introduce as in  \cite{Borisov:09::}
 a new independent  variable $\tau$ such that
 $\mathrm{d}\tau/\mathrm{d}t=1/r^2$. Then 
we find that 
\[
p_r=\dfrac{r'}{r^2},\qquad p_{\varphi}=\varphi',
\]
where prime denotes the differentiation with respect to $\tau$.
In effect we have 
\[
H=\dfrac{r'^2}{2r^4}+\dfrac{1}{r^2}G,\mtext{and}
G=\dfrac{\varphi'^2}{2}+\dfrac{a }{\cos^2(n\varphi)},
\]
i.e., we  effectively separated variables
\begin{equation}
 \int \dfrac{\mathrm{d}r}{r\sqrt{2(Hr^2-G)}}=\tau +C_1,\qquad
 \int
 \dfrac{\cos(n\varphi)\mathrm{d}\varphi}{\sqrt{2(G\cos^2(n\varphi)-a )}}=
\tau +C_2.
  \label{eq:ccalki0}
\end{equation}
The explicit forms of these elementary integrals are following
\begin{equation}
  \dfrac{1}{\sqrt{2G}}\arctan\sqrt{\dfrac{Hr^2-G}{G}}=\tau+C_1,
  \label{eq:sum10}
\end{equation}
and
\begin{equation}
  \dfrac{1}{n\sqrt{2G}}\arcsin\left[\sqrt{\dfrac{G}{G-a }}\sin(n\varphi)\right]=\tau+C_2.
  \label{eq:sum20}
\end{equation}
From  \eqref{eq:sum10} and \eqref{eq:sum20} we deduce  that
\begin{equation}
  I=n\sqrt{2G}(C_2-C_1)=\arcsin\left[\sqrt{\dfrac{G}{G-a }}\sin(n\varphi)\right]- 
  n\arctan\sqrt{\dfrac{Hr^2-G}{G}},
\end{equation} 
is a first integral of the system.   Using it we find an algebraic
first integral. To this end we  perform a sequence of transformations
applying the following   formulae
\begin{equation*}
  \arcsin z=-\rmi\ln\left(\rmi z+\sqrt{1-z^2}\right),\,\,
  \arccos z=-\rmi\ln\left( z+\sqrt{z^2-1}\right),\,\,
  \arctan z=\dfrac{\rmi}{2}\ln\left(\dfrac{1-\rmi z}{1+\rmi z}\right).
\end{equation*} 
Using them and making some simplifications  we obtain
\[
I=-\rmi\ln\left\{\left(\dfrac{p_{\varphi}\cos(n\varphi)}{\sqrt{2(G-a )}}+
    \rmi\sqrt{\frac{G}{G-a }}\sin(n\varphi)\right)\dfrac{(\sqrt{2G}-\rmi
    r p_r)^n}{(2H)^{n/2}r^n} \right\}.
\]
From the above formula we deduce that 
\begin{equation}
  \begin{split}
    I_{1}&=(2H)^{n/2}\sqrt{2(G-a )}\exp(\rmi \,I) \\
            &=\dfrac{1}{r^n}\left(p_{\varphi}\cos(n\varphi)+\rmi\sqrt{2G}\sin(n\varphi)\right)(\sqrt{2G}-\rmi r p_r)^n,
  \end{split}
\label{eq:brumbrum}
\end{equation} 
is a first integral of the system.
For rational $n$ this integral
is an algebraic function of Cartesian variables $(q_1,q_2,p_1,p_2)$.  

If the considered system is real,  then one would like to possess real first
integrals. Taking  the real and imaginary parts of $I_1$ (assuming
that all variables are real) we obtain real first integrals. Let us
assume for simplicity that  $n$ is a positive integer. Then 
\begin{equation*}
I_1= r^{-n}\left(p_{\varphi}\cos(n\varphi)+\rmi\sqrt{2G}\sin(n\varphi)\right)
    \sum_{k=0}^n\binom{n}{k}(2G)^{(n-k)/2}(-\rmi)^kr^kp_r^k, 
\end{equation*}
 and from this we obtain   
\begin{equation}
  \begin{split}
    &F_1=\Re
    I_1=\sum_{k=0}^{[n/2]}(-1)^k\binom{n}{2k}(2G)^{\frac{n-2k}{2}}
    \frac{p_r^{2k}}{r^{n-2k}}\left[p_{\varphi}\cos(n\varphi)+
      \dfrac{n-2k}{2k+1}rp_r\sin(n\varphi)\right],\\
    &F_2=\Im
    I_1=\sum_{k=0}^{[n/2]}(-1)^k\binom{n}{2k}(2G)^{\frac{n-2k-1}{2}}
\frac{p_r^{2k}}{r^{n-2k}}\left[2G\sin(n\varphi)-
      \frac{n-2k}{2k+1}rp_rp_{\varphi}\cos(n\varphi)\right].
  \end{split}
\end{equation} 
Here $[x]$ denotes the integer part of $x$.
We note that always one of these first integrals is a polynomial in
momenta $(p_r,p_{\varphi})$. For $n$ even expression $(2G)^{(n-2k)/2}$
is a polynomial and as result $F_1$ is a polynomial in the
momenta. Similarly one can deduce that for $n$ odd $F_2$ is a
polynomial in the momenta.  Let us note that  if we put negative $n$ in  \eqref{eq:hamidlo}, then the potential 
does not change,
thus we can assume that always $n>0$. The same is true also for more general form of potential ~\eqref{eq:V0}.

For positive rational $n=n_1/n_2$ from \eqref{eq:brumbrum} also a polynomial in
the momenta first integral can be constructed. Namely we consider the
new first integral 
\[
I_2:=I_1^{n_2}=r^{-n_1}\left(p_{\varphi}\cos(n\varphi)+\rmi\sqrt{2G}\sin(n\varphi)\right)^{n_2}(\sqrt{2G}-\rmi r p_r)^{n_1}.
\]
Separating real and imaginary parts of this first integral we find
that for  $n_1$ even and $n_2$ odd integral $F_1:=\Re I_2$ is
polynomial in momenta. Moreover,  for odd $n_1$ integral   $F_2:=\Im I_2$ is
polynomial in momenta independently of the parity of $n_2$.  

The described  direct  approach works  perfectly in the same way for
the general form of the potential~\eqref{eq:V0} and it gives the
following  form of the first integral 
\begin{equation}
\label{eq:gI1}
 I_1=r^{-2n}(\sqrt{2G}-\rmi
 rp_r)^{2n}\left[\sqrt{2G}\sin(2n\varphi)p_{\varphi}-2\rmi(G\cos(2n\varphi)+b-a)\right].
\end{equation}
Assuming that $n$ is a positive integer, then 
\begin{equation}
I_1=
    r^{-2n}\left[\sqrt{2G}\sin(2n\varphi)p_{\varphi}-2\rmi(G\cos(2n\varphi)+b-a)\right]
      \sum_{k=0}^{2n}\binom{2n}{k}
    (2G)^{(2n-k)/2}(-\rmi)^kr^kp_r^k,
  \label{eq:firstinto}
\end{equation} 
and the real and imaginary parts of this complex function give additional
first integrals
\[
\begin{split}
  F_1=&\sum_{k=0}^n(-1)^k\binom{2n}{2k}(2G)^{n-k}
\frac{p_r^{2k}}{r^{2(n-k)}}\left[G\sin(2n\varphi)p_{\varphi}-\frac{2(n-k)}{2k+1}
    (G\cos(2n\varphi)+b-a)rp_r\right],\\
  F_2=&\sum_{k=0}^n(-1)^k\binom{2n}{2k}(2G)^{n-k} \frac{p_r^{2k}}{r^{2(n-k)}}\left[\frac{2(n-k)}{2k+1}\sin(2n\varphi)rp_rp_{\varphi}
    +2(G\cos(2n\varphi)+b-a)\right].
\end{split}
\]
In the above formulae $F_1=\Re(I_1)/2\sqrt{2G}$ and
$F_2=\Im(I_1)$. Proceeding in the way similar to the previous case we
can also construct polynomial in the momenta first integrals for
positive rational $n$.

Obtained results can be rewritten immediately for Hamiltonian systems
with indefinite flat form of kinetic energy,  which in polar
coordinates are given by the
following Hamilton function 
\begin{equation}
  H=\dfrac{1}{2}\left(p_r^2-  \frac{p_\varphi^2}{r^2}\right)+V,
\end{equation} 
with potential
\begin{equation}
    V=\dfrac{a }{r^2\cosh^2(n\varphi)}+\dfrac{b }{r^2\sinh^2(n\varphi)}.
\end{equation} 
Coordinates $(r,\varphi)$ are related to the Cartesian coordinates
by the formulae
\begin{equation}
  q_1=r\cosh(\varphi),\qquad q_2=r\sinh(\varphi).
  \label{eq:trans}
\end{equation}
This system is separable   in $(r,\varphi)$ coordinates with first
integral 
\begin{equation}
\label{eq:ggcurved}
G= \frac{1}{2} p_{\varphi}^2 -\dfrac{a }{\cosh^2(n\varphi)}-\dfrac{b }{\sinh^2(n\varphi)}.
\end{equation}

One more additional first integral has the form      
\begin{equation}
I_1=r^{-2n}(\sqrt{2G}-
rp_r)^{2n}\left[\sqrt{2G}\sinh(2n\varphi)p_{\varphi}+2(G\cosh(2n\varphi)+a +b )\right].
\label{eq:ii1}
\end{equation}
It can be obtained either  by a direct integration, or from
integral~\eqref{eq:gI1} by substitutions
\begin{equation}
\varphi\to \rmi \varphi, \qquad
 p_{\varphi}\to -\rmi p_{\varphi} \quad G\to -G,\quad \sqrt{G}\to \rmi\sqrt{G},\quad b\to-b.
\label{eq:substy}
\end{equation}
One can construct also another first integral
\begin{equation}
 I_2=r^{-2n}(\sqrt{2G}+
rp_r)^{2n}\left[\sqrt{2G}\sinh(2n\varphi)p_{\varphi}-2(G\cosh(2n\varphi)+a +b )\right]
\label{eq:ii2}
\end{equation}
from integral~\eqref{eq:gI1} choosing the other square root of $-G$,
i.e. making the substitution $\sqrt{G}\to -\rmi\sqrt{G}$ in ~\eqref{eq:gI1}. Then, for 
$n\in\N$, either $I_1+I_2$, or  $I_1-I_2$, is polynomial in momenta
first integral.  In  general case for  positive rational $n=n_1/n_2$,
either   
$F_1:=I_1^{n_2}+I_2^{n_2}$, or $F_1:=I_1^{n_2}+I_2^{n_2}$ is a first
integral which is polynomial in momenta.

\subsection{Other  form of the additional integral, polynomial in momenta}
In the previous section we showed that the additional first integral is polynomial in polar
momenta $p_{\varphi}$ and $p_r$. Here we show an approach which allows to demonstrate that this
integral is expressible in terms of polynomials closely related with the Chebyshev polynomials.
The obtained form of the first integral shows that is rational in Cartesian variables
$(q_1, q_2, p_1, p_2)$ and polynomial in momenta $(p_1, p_2)$.

Let us introduce double  spherical coordinates 
\begin{equation}
 q_1=r\cos\varphi,\quad q_2=r\sin\varphi,\qquad  p_1=p\cos\psi,\quad p_2=p\sin\psi.
\label{eq:polar}
\end{equation} 
Let us consider for example natural   Hamiltonian with potential \eqref{eq:complet} for $a=1$ and $b=0$. In polar coordinates 
\eqref{eq:polar} Hamiltonian takes the form
\[
 H=\dfrac{1}{2}p^2+\dfrac{1}{r^2\cos^2(n\varphi)},
\]
and Hamiltonian equations transform into
\begin{equation}
 \begin{split}
\dot r&=p\cos(\varphi-\psi),\\
\dot \varphi&=-\dfrac{p}{r}\sin(\varphi-\psi),\\
\dot p&=-\dfrac{2}{r^3\cos^3(n\varphi)}\left[\dfrac{n-1}{2}\cos\left((n+1)\varphi-\psi\right)-\dfrac{n+1}{2}\cos\left((n-1)\varphi+\psi\right)\right],\\
\dot \psi&=-\dfrac{2}{pr^3\cos^3(n\varphi)}\left[\dfrac{n-1}{2}\sin\left((n+1)\varphi-\psi\right)+\dfrac{n+1}{2}\sin\left((n-1)\varphi+\psi\right)\right].
 \end{split}
\label{eq:hamcyl}
\end{equation} 
Let us note that transformation \eqref{eq:polar} is not canonical. In these coordinates 
Jacobi first integral takes the form
\[
 I_0=r^2p^2\sin^2(\varphi-\psi)+\dfrac{2}{\cos^2(n\varphi)}.
\]
Let us look for an additional first integral of the form
\[
I=p^n\sin(n\psi)+\sum_{i=1}^{[n/2]}(-2)^i\dfrac{p^{n-2i}}{r^{2i}\cos^{2i}(n\varphi)}\sum_{m=i-1}^{n-1-i}a_{i,m}\sin[2(n-m-1)\varphi-(n-2m-2)\psi],
\]
where $a_{im}$ are unknown
constant coefficients. Substitution of these formulae into condition $\dot I=0$ yields the following recurrence equation on $a_{im}$
\begin{equation}
\begin{split}
& [i(n-1)+m+1]a_{i+1,m+1}-[i(n+1)+m+2]a_{i+1,m}=(n+1)(n-m-i-1)a_{i,m}\\
&-(n-1)(m-i+2)a_{i,m+1},\qquad i=1,\ldots,\left[\dfrac{n-2}{2}\right],\quad m=i,\ldots,n-i-2.
\end{split}
\end{equation} 
It has the following solution
\begin{equation}
 a_{i,m}=\dfrac{(m+2-i)_i(n-m-i)_{i-1}}{(1)_{i-1}(2)_{i-1}},
\end{equation} 
where 
\[
 (a)_n=a (a+1) \cdots (a+n-1)
\]
is the Pochhammer symbol.

We define homogeneous polynomials $f_n$ and $g_n$ and $F_{n}$,
$G_{n}$ in the following way
\[
 (q_1+\rmi q_2)^n=f_n+\rmi g_n,\qquad  (p_1+\rmi p_2)^n=F_n+\rmi G_n.
\]
From this definition the connection of polynomials $f_n,g_n$ as well as $F_n,G_n$ with Chebyshev polynomials is obvious,
see e.g. Section~2 in \cite{Freudenburg:09::}.

Then the first integral  $I$ in the Cartesian coordinates has the 
form 
\[
I=G_{n}+\sum_{i=1}^{[n/2]}\dfrac{(-2)^i}{f_{n}^{2i}}\sum_{m=i-1}^{n-1-i}a_{i,m}r^{2[(i-1)(n-1)+m]}p^{2(m-i+2)}S_{i,m},
\]
where 
\[
S_{i,m}:=\Im\left[(p_{1}-\rmi p_{2})^{n-2m-2}(q_{1}+\rmi
  q_{2})^{2(n-m-1)}\right]. 
\]
Let us notice that 
\[
\Im\left[(p_{1}-\rmi p_{2})^{\alpha}(q_{1}+\rmi
  q_{2})^{\beta}\right]=F_{\alpha}g_{\beta}-G_{\alpha}f_{\beta}.
\]
Hence the above first integral is rational in $(q_1,q_2,p_1,p_2)$  and
polynomial in $(p_1,p_2)$.

In the similar way one can treat the general potential
\eqref{eq:complet}, however in this case the form of the first
integral is much more complicated.

\subsection{Generalisation of TTW system to a  system on $\bbS^2$ and $\bbH^2$}
In this section we consider systems given by the Hamiltonian
function~\eqref{eq:hamkappa} with potential   
\begin{equation}
\label{eq:vv}
V^{(\kappa)}_{n}(r, \varphi):= \frac{1}{\sk^2(r)}U(\varphi), 
\end{equation}
where 
\begin{equation*}
 U(\varphi)=\dfrac{a }{\cos^2(n\varphi)}+\dfrac{b
 }{\sin^2(n\varphi)}. 
\end{equation*}
This is a natural generalisation of the system considered in the
previous subsection onto the spaces with a constant non-zero
curvature\footnote{During the final stage of preparation the
   work of \cite{Kalnins:arXiv1006.0864} appeared where the reader will
  find other examples of super-integrable systems on constant
  curvature spaces.}.  It is not difficult to show that for this
potential~Lemma~\ref{lem:1} applies. That is, if the potential is
super-integrable, then $n$ is non-zero rational number. We show that
in fact those potentials are super-integrable, i.e., that the
necessary conditions of Theorem~\ref{thm:we} are also sufficient.  To
this end it is enough to perform the explicit integration similar to
that done in the previous subsections. It gives the following form of
the first integral
\begin{equation}
 I=\left(\sqrt{2G}\,\frac{\ck(r)}{\sk(r)}+\rmi p_r\right)^{2n}\left(\sqrt{G}p_{\varphi}\sin(2n\varphi)
+\rmi\sqrt{2}(G\cos(2n\varphi)+
b-a )\right),
\label{eq:brumfr}
\end{equation}
 and the first integral $G$ takes the form 
\begin{equation*}
G=\frac{1}{2}p_{\varphi}^2+U(\varphi).
\end{equation*}
Assuming that $n\in\N^{\star}:=\N\setminus\{0\}$, $a,b\in\R$, and taking    real and imaginary
parts of \eqref{eq:brumfr} we obtain  the following explicit  forms of first integrals
\[
 \begin{split}
I_1&=\sqrt{G}\Re I=  \sum_{j=0}^n(-1)^j\binom{2n}{2j}(2G)^{n-j}\left(\frac{\ck(r)}{\sk(r)}\right)^{2n-2j-1}p_r^{2j}
\left[Gp_{\varphi}\sin(2n\varphi)\frac{\ck(r)}{\sk(r)}\right.\\
&\left.-\frac{2(n-j)}{2j+1}p_r(G\cos(2n\varphi)+b-a )\right],\\
I_2&=\frac{\Im I}{\sqrt{2}}=\sum_{j=0}^n(-1)^j\binom{2n}{2j}(2G)^{n-j}\left(\frac{\ck(r)}{\sk(r)}\right)^{2n-2j-1}p_r^{2j}
\left[\frac{\ck(r)}{\sk(r)}(G\cos(2n\varphi)+b-a )\right.\\
&\left.+\frac{n-j}{2j+1}p_rp_{\varphi}\sin(2n\varphi)\right].
 \end{split}
\]
  
Analogous calculations can be repeated for potential
\[
 U(\varphi)=\dfrac{a }{\cosh^2(n\varphi)}+\dfrac{b }{\sinh^2(n\varphi)}.
\]

The Jacobi first integral for it takes the form \eqref{eq:ggcurved} and this one obtained from separation of variables in 
polar coordinates is
\[
\begin{split}
 I&=\left(\sqrt{2G}\,\frac{\ck(r)}{\sk(r)}+ p_r\right)^{2n}\left(-\sqrt{G}p_{\varphi}\sinh(2n\varphi)+\sqrt{2}(G\cosh(2n\varphi)+
a+b )\right)\\
&=\left(-\sqrt{G}p_{\varphi}\sinh(2n\varphi)+\sqrt{2}(G\cosh(2n\varphi)+
a+b )\right)\sum_{j=0}^{2n}\binom{2n}{j}(2G)^{(2n-j)/2}\left(\frac{\ck(r)}{\sk(r)}\right)^{2n-j}p_r^j.
\end{split}
\]
It can be also obtained directly from \eqref{eq:brumfr} using the substitutions \eqref{eq:substy}. One can also construct first integral
using different root of $-G$ i.e. making the substitution $\sqrt{G}\to -\rmi\sqrt{G}$  in \eqref{eq:brumfr}  and we obtain the following form
\[
 I=\left(\sqrt{G}p_{\varphi}\sinh(2n\varphi)+\sqrt{2}(G\cosh(2n\varphi)+
a+b )\right)\sum_{j=0}^{2n}(-1)^j\binom{2n}{j}(2G)^{(2n-j)/2}\left(\frac{\ck(r)}{\sk(r)}\right)^{2n-j}p_r^j.
\]
Proceeding in the way similar to the previous cases we
can also construct polynomial in the momenta first integrals for positive rational $n$.
\section*{Acknowledgements}
The authors  are very thankful to anonymous referees for their remarks, comments and suggestions that allowed to improve considerably
this paper. 

The first two authors are very grateful to  Alexey V. Borisov, Ivan S. Mamaev and Alexander A. Kilin for valuable 
 discussions during their visit  in Izhevsk.

This research has been partially supported by grant No. N N202 2126 33
of Ministry of Science and Higher Education of Poland and by EU
funding for the Marie-Curie Research Training Network AstroNet.

The third author's work was partially supported by the
Grant-in-Aid for Scientific Research of Japan Society for the
Promotion of Sciences (JSPS), No.18540226.

\bibliographystyle{plainnat}
\newcommand{\noopsort}[1]{}\def\cprime{$'$} \def\cydot{\leavevmode\raise.4ex\hbox{.}}

\end{document}